\documentclass[12pt, onecolumn, draftclsnofoot,peerreview]{IEEEtran}
\usepackage{times}
\usepackage[final]{graphicx}
\usepackage[reqno]{amsmath}
\usepackage{amsfonts}

\usepackage{epsfig}
\usepackage{amssymb}

\newtheorem{proposition}{Proposition}
\newtheorem{theorem}{Theorem}
\newtheorem{lemma}{Lemma}
\newtheorem{definition}{Definition}
\newtheorem{corollary}{Corollary}

\DeclareMathAlphabet{\mathpzc}{OT1}{pzc}{m}{it}

\begin{document}

\title{Fingerprinting with Minimum Distance Decoding}

\author{
\authorblockN{Shih-Chun Lin, Mohammad Shahmohammadi and Hesham El Gamal*}
\thanks{S. C. Lin is with Department of Electrical Engineering,
National Taiwan University, Taipei, Taiwan 10617. The work of S.
C. Lin was supported by ``Graduate Students Study Abroad Program''
of National Science Council, Taiwan, R.O.C.}
\thanks{M. Shahmohammadi and H. El Gamal are with Department of
Electrical and Computer Engineering, The Ohio State University,
Columbus, OH, 43210. This work was partly performed while Hesham
El Gamal was visiting Nile University, Cairo, Egypt. The authors
acknowledge the generous funding of the National Science
Foundation, USA
 }
\thanks{E-mail: \{lins, shahmohm, helgamal\}@ece.osu.edu.} }

\maketitle \IEEEpeerreviewmaketitle

\begin{abstract}
This work adopts an information theoretic framework for the design
of collusion-resistant coding/decoding schemes for digital
fingerprinting. More specifically, the minimum distance decision
rule is used to identify $1$ out of $t$ pirates. Achievable rates,
under this detection rule, are characterized in two distinct
scenarios. First, we consider the averaging attack where a random
coding argument is used to show that the rate $1/2$ is achievable
with $t=2$ pirates. Our study is then extended to the general case
of arbitrary $t$ highlighting the underlying complexity-performance
tradeoff. Overall, these results establish the significant
performance gains offered by minimum distance decoding as compared
to other approaches based on orthogonal codes and correlation
detectors which can support only a sub-exponential number of users
(i.e., a zero rate). In the second scenario, we characterize the
achievable rates, with minimum distance decoding, under any
collusion attack that satisfies the marking assumption. For $t=2$
pirates, we show that the rate $1-H(0.25)\approx 0.188$ is
achievable using an ensemble of random linear codes. For $t\geq 3$,
the existence of a {\em non-resolvable} collusion attack, with
minimum distance decoding, for any non-zero rate is established.
Inspired by our theoretical analysis, we then construct
coding/decoding schemes for fingerprinting based on the celebrated
Belief-Propagation framework. Using an explicit repeat-accumulate
code, we obtain a vanishingly small probability of misidentification
at rate $1/3$ under averaging attack with $t=2$. For collusion
attacks which satisfy the marking assumption, we use a more
sophisticated accumulate repeat accumulate code to obtain a
vanishingly small misidentification probability at rate $1/9$ with
$t=2$. These results represent a marked improvement over the best
available designs in the literature.
\end{abstract}

\vspace{-4mm}
\begin{center}
   {\underline{\bf \small EDICS}} \hspace{3mm} {\small WAT-FING}
\end{center}
\vspace{-6mm}

\section{Introduction}
Digital fingerprinting is a paradigm for protecting copyrighted data
against illegal distribution \cite{Boneh}. In a nutshell, a
\emph{distributor}, i.e., the provider of copyrighted data, wishes
to distribute its data $\mathbb{D}$ among a number of licensed
\emph{users}. Each licensed copy is identified with a mark, which
will be referred to as a \emph{fingerprint} in the sequel, composed
of a set of redundant digits embedded inside the copyrighted data.
The locations of the redundant digits are kept \emph{hidden} from
the users and are only known to the distributor. Their positions,
however, remain the same for all users. If any user re-distributes
its data in an unauthorized manner, it will be easily identified by
its fingerprint. However, several users may collude to form a
\emph{coalition} enabling them to produce an unauthorized copy which
is difficult to trace. In the literature, the colluding members are
typically referred to as \emph{pirates} or \emph{colluders}. Hence,
the need arises for the design of collusion-resistant digital
fingerprinting techniques. Our work develops an information
theoretic framework for the design of low complexity {\em
pirate-identification} schemes.

To enable a succinct development of our results, we first consider
the widely studied \emph{averaging attack}~\cite{Liu}. The
colluders, in this strategy, average their media contents to
produce the forged copy. An explicit fingerprinting code
construction for this attack was proposed in \cite{Liu}. In this
construction, however, the maximum number of users $M$, grows only
polynomially with the fingerprinting code-length $n$ (more
precisely $M=O(n^2)$). Clearly, this rate of growth corresponds to
a zero rate in the information theoretic sense. This motivates our
pursuit for a fingerprinting scheme which supports an
exponentially growing number of users, with the code-length, while
allowing for low complexity pirate-identification strategies.
Towards this goal, we use a random coding argument to establish
the existence of a rate $0.5$ {\em linear} fingerprinting code
which achieves a vanishingly small probability of
misidentification when 1) Only $t=2$ pirates are involved in the
averaging attack and 2) The low complexity minimum distance (MD)
decoder is used to identify one of the two pirates. The enabling
observation is the intimate connection between the scenario under
consideration and the binary erasure channel (BEC). This result is
then extended to the general case with an arbitrary coalition size
$t$ where the tradeoff between complexity and performance is
highlighted.

Building on our analysis for the averaging attack, we then proceed
to fingerprinting strategies which are resistant to more general
forging techniques. More specifically, we adopt the \emph{marking
assumption} first proposed in~\cite{Boneh}. In this framework, the
pirates attempt to identify the positions occupied by the
fingerprinting digits by comparing their copies. Afterwards, they
can {\em only} modify the identified coordinates, in any desired
way, to minimize the probability of traceability. The validity of
the marking assumption hinges on the assumption that any
modification to the data content $\mathbb{D}$ will damage it
permanently. This prevents the users from modifying any location
in which they do not identify as a fingerprinting digit since it
{\em may be} a data symbol. Boneh and Shaw~\cite{Boneh} were the
first to construct fingerprinting codes that are resistant to
attacks that satisfy the marking assumption. This approach was
later extended in \cite{Barg_code} using the idea of separating
codes \cite{Spt}. To the best of our knowledge, the best available
explicit binary fingerprinting codes are the {\em low rate} codes
presented in~\cite{Barg_code}. For example, for $t=2$, the best
available code has a rate$\approx 0.0092$. More recently, upper
and lower bounds on the binary fingerprinting capacity for $t=2$
and $t=3$ were derived in~\cite{Barg_cap}. The decoder used in
\cite{Barg_cap}, however, was based on exhaustive search, and
hence, would suffer from an exponentially growing complexity in
the code length. This prohibitive complexity motivates our
proposed approach. In this paper, we show that using linear
fingerprinting codes and MD decoding, one can achieve rates less
than $0.188$ when the coalition size is $t=2$. Unfortunately, the
proposed approach does not scale for $t\geq 3$. This negative
result calls for a more sophisticated identification technique
inspired by the analogy between our set-up and multiple access
channels. Our results in this regard will be reported elsewhere.

Since the complexity of the {\em exact} MD decoder can be
prohibitive when the code-length is long, we develop a low
complexity belief-propagation (BP) identification approach
\cite{richardson2001cld}\cite{pishronik2004dld}. This detector
only requires a linear complexity in $n$, and offer remarkable
performance gain over the best known explicit constructions for
fingerprinting \cite{Barg_code}\cite{Liu}. For example, we propose
a modified iterative decoder tailored for the averaging attack
with $t=2$. Using this decoder along with an explicit
repeat-accumulate (RA) fingerprinting code, we achieve a
vanishingly small probability of misidentification for rates up to
$1/3$. For the marking assumption set-up, we achieve a vanishingly
small misidentification probability for rates up to $1/9$ using
the recently proposed class of low rate accumulate repeat
accumulate (ARA) codes~\cite{divsalar}. It is worth noting that
these results represent a marked improvement over the state of the
art in the literature. Furthermore, one would expect additional
performance enhancement by optimizing the degree sequences of the
codes (which is beyond the scope of this work).

The rest of the paper is organized as follows. In
Section~\ref{Sec_Not}, we introduce the mathematical notations and
formally define our problem setup. Then we explore the theoretical
limits of fingerprinting using the MD decoder in
Sections~\ref{Sec_Theorem}~and~\ref{Sec_Theorem_2}. The simulation
results based on the BP framework are presented in
Section~\ref{Sec_BP}. Finally, Section~\ref{Sec_conclu} offers some
concluding remarks.

\section{Notations and Problem statement} \label{Sec_Not}
Throughout the paper, random variables and their realizations are
denoted by capital letters and corresponding smaller case letters,
respectively. Deterministic vectors are denoted by bold-face
letters. We denote the entropy function by $H(\cdot)$, with the
argument being the probability mass function. Furthermore, for
simplicity, we abbreviate $H(p,1-p)$ by $H(p)$, where $1 \geq p
\geq 0$. For two functions of $n$, we write $a(n)\doteq b(n)$ if:
$\lim \limits_{n\rightarrow\infty }\frac{1}{n}\frac{a(n)}{b(n)}=
1$, for example, ${n \choose d}\doteq 2^{nH(\frac{d}{n})}$. The
Hamming distance between two vectors $\mathbf{x}_1,\mathbf{x}_2$
is denoted by $d_H(\mathbf{x}_1,\mathbf{x}_2)$. Without loss of
generality, we assume that the number of users is $M$, and hence,
a coalition $U$ of size $t$ is a subset of $\{1,2,\ldots,M\}$
where $|U| = t$. The goal of the coalition, in a nutshell, is to
produce a forged fingerprint, $\mathbf{y}$, such that the
distributor will not be able to trace it back to any of its
members. In the following, we first introduce the notation that
will be used for a general attack satisfying the marking
assumption and then specify our notations for the averaging attack
scenario. It should be noted that our formulation follows in the
footsteps of \cite{Barg_cap}. For completeness, however, we repeat
it here. As mentioned in \cite{Boneh}, deterministic
fingerprinting under the marking assumption is not possible in
general. Therefore, the distributor needs to employ some kind of
randomization which leads to a collection of binary codes $(F,G)$
composed of $K$ pairs of encoding and decoding functions as:
\begin{eqnarray}
& & f_k:\{1,2,\ldots,M\}\rightarrow \{0,1\}^n \\
& & g_k:\{0,1\}^n\rightarrow\{1,2,\ldots,M\} \nonumber \\
& & k = 1,2,\ldots,K, \nonumber
\end{eqnarray}
where the code rate $R$ is $\frac{\log_2 M}{n}$ and the secret key,
$k$ is a random variable employed to randomize the codebook. This
way, the exact codebook utilized for fingerprinting is kept hidden
from the users. It should be noted that, adhering to common
conventions in cryptography, the family of encoding and decoding
functions as well as the probability distribution of the secret key,
$p(k)$, are known to all users. Finally, it is clear from the
definition of $g_k$ that the objective of the distributer, in our
formulation, is to identify only one of the colluders correctly.

For simplicity of presentation, let's assume that $t=2$ then the
fingerprints corresponding to the coalition of users (also referred
to as pirates or colluders), $u_1,u_2$ are denoted by
$\{\mathbf{x}_1,\mathbf{x}_2\}$. The marking assumption implies that
position $i$ is \emph{undetectable} to the two colluders if
$x_{1i}=x_{2i}$, otherwise it is called
\emph{detectable}~\cite{Boneh}. Those undetectable coordinates can
not be changed by the pirates, and hence, the set of all possible
forged copies is give by

\begin{equation}
E(U) = \{\mathbf{y}\in \{0,1\}^n\mid y_i=x_{1i},\forall{i}\quad
undetectable  \}.
\end{equation}
In general, a coalition $U$ may utilize a random strategy that
satisfies the marking assumption to produce $\mathbf{y}$. That is,
if ${V(\mathbf{y}\mid \mathbf{x}_1,\mathbf{x}_2)}$ is the
probability that $\mathbf{y}$ is created, given the coalition
$\{\mathbf{x}_1,\mathbf{x}_2\}$, then we have:
\begin{equation}
{V(\mathbf{y}\mid \mathbf{x}_1,\mathbf{x}_2)} = 0 \qquad \textrm{for
all } \mathbf{y}\not\in E(U). \\ \label{Eq_Adm}
\end{equation}
In this paper, we focus on the maximum probability of
misidentification over the set of all strategies which satisfy
(\ref{Eq_Adm}) (denoted by $\mathcal{V}$ in the sequel). Similar
to~\cite{Barg_cap}, we average the probability of misidentification
over all possible coalitions leading to the following performance
metric:
\begin{equation} \label{Pm_Def}
\overline{P}_m(F,G):=\frac{1}{{M \choose t}}\sum_{U}{\max_{V \in
\mathcal{V}}P_m(U,F,G,V)},
\end{equation}
where
\[
P_m(U,F,G,V) := \mathbb{E}_K \Big( \sum_{\mathbf{y}\in
E(U),g_k(\mathbf{y})\notin U} V(\mathbf{y}\mid f_k(U))\Big ).
\]
In the case of an averaging attack, we employ the typical assumption
of mapping the binary fingerprints into the antipodal alphabets
$\{-1,1\}$ where the encoder now is defined as~\cite{Liu}
\begin{equation} 
f:\{1,2,\ldots,M\}\rightarrow \{-1,+1\}^n.
\end{equation}
As anticipated from the name, the forged copy is now given by:
\begin{equation} \label{Eq_avg_attack}
\mathbf{y}=\frac{1}{t}\sum_{i=1}^t\mathbf{x}_i,
\end{equation}
where the addition is over real field. The decoder is now defined as
\begin{equation} \label{Eq_g_avg}
g:\{\mathpzc{A}_\mathbf{y}\}^n\rightarrow\{1,2,\ldots,M\},
\end{equation}
where $\mathpzc{A}_\mathbf{y}$ is the alphabets of $\mathbf{y}$, for
example, it is $\{-1,0,+1\}$ when $t=2$. Misidentification will
happen if $g(\mathbf{y}) \notin U$. Note that for $t=2$, if
$g(\mathbf{y}) \in U$, i.e., we trace one colluder correctly then we
can always trace another colluder correctly according to
(\ref{Eq_avg_attack}). In this special case, the performance metric
in (\ref{Pm_Def}) reduces to
\begin{equation}
\overline{P}^a_m:=\frac{1}{{M \choose t}} \sum_{U}(g(\mathbf{y})
\notin U).
\end{equation}

\section{The Averaging Attack} \label{Sec_Theorem} In this section, we investigate the
theoretical achievable rate of fingerprinting code with the minimum
distance (MD) decoder under the averaging attack. First, we need the
following definition.
\begin{definition} \label{def_capacity}
We say that the capacity of of an ensemble of fingerprinting
codebooks ${\cal E}$  is $R_{\cal E}$ under MD decoding if

\begin{enumerate}

\item For $M=2^{nR}$ with $R<R_{\cal E}$, the average probability of misidentification over the ensemble $P_m$
using MD decoding goes exponentially to zero as the codelength $n$
goes to infinity.
\item Conversely, for $M=2^{nR}$ with $R>R_{\cal
E}$, there exists a constant $\delta>0$ such that $P_m>\delta$ for
sufficiently large block lengths.
\end{enumerate}
\end{definition}
Note that this converse in the previous definition is applicable
only to a specific family of codes similar to the approach taken
in~\cite{richardson2001cld,pishronik2004dld}. We also call a rate
is MD-achievable if only the first part in Definition
\ref{def_capacity} is met. We are now ready to prove our first result. \\

\begin{theorem} \label{Theorem_Avg_iid}
The fingerprinting capacity of the i.i.d codebook ensemble when
$t=2$ is $R_{\cal E}=0.5$ (under the averaging attack and the MD
decoder).
\\
\end{theorem}
\begin{proof} The encoder and decoder come as follows. \\
\noindent \textbf{Encoder:} The encoder chooses codewords
uniformly and independently from all $2^n$ different vectors
belonging to $\{0,1\}^n$, transfers the fingerprinting codeword
alphabets from $\{0,1\}$ to $\{-1,+1\}$, and assigns the
fingerprints to the users.
\\ \\
\noindent \textbf{Decoder:} With the given forged fingerprint
$\mathbf{y}$, the decoder treats the position $i$ where
$\mathbf{y}_i=0$ as an erased position, and the others as unerased
positions. Let $\mathpzc{E}$ be the set of erasure positions and
$\overline{\mathpzc{E}} := [1:n] \setminus \mathpzc{E}$. Also let
$\mathbf{y}_{\overline{\mathpzc{E}}}$ denote those components of
$\mathbf{y}$ which are indexed by $\overline{\mathpzc{E}}$. The
decoder will search the codebook to find the codeword which agrees
with $\mathbf{y}$ in all unerased positions
$\mathbf{y}_{\overline{\mathpzc{E}}}$. Once the decoder finds such
a codeword, the decoder declares it as the pirate. A
misidentification occurs when the codeword of an innocent user
$\mathbf{z}$ is consistent with $\mathbf{y}$.
\\
Achievability: For a small $\varepsilon$, we say the assigned
fingerprints $\mathbf{x}_1,\mathbf{x}_2$ are \emph{close} if
$d_H(\mathbf{x}_1,\mathbf{x}_2)\leq n(\frac{1}{2}+\varepsilon)$,
here the fingerprinting alphabets are $\{0,1\}$ before
transformation. As shown in Appendix \ref{App_close_pair}, we know
that with high probability, $(\mathbf{x}_1,\mathbf{x}_2)$ are a
close pair. Thus, given a small $\epsilon>0$,
\begin{equation} 
|\mathpzc{E}| \leq n(\frac{1}{2}+\epsilon),
\end{equation}
since the erasures happen when the bits of $(\mathbf{x}_1,\mathbf{x}_2)$ are different.
For the given forged fingerprint $\mathbf{y}$, $\mathbf{z}$ must
agree with $\mathbf{y}$ in all $n-|\mathpzc{E}|$ unerased
positions, and can be $-1$ or $+1$ in the rest $|\mathpzc{E}|$
erased positions. The probability of choosing such codeword is
upper-bounded by
\begin{equation} 
2^{n*(1/2+\epsilon)}/2^n.
\end{equation}
By using the union bound, we know that for $R<1/2-\epsilon$, the
probability of misidentification $P_m$ tends to zero exponentially
fast for sufficiently large
codeword length $n$.
\\ \\
Converse: From (\ref{Eq_random_dis}) in the Appendix, we know that
$P(|\mathpzc{E}| \geq n/2)>P(|\mathpzc{E}|=n/2)=\delta$, where
$\delta$ is non-vanishing with respect to codeword length $n$. For a
fingerprinting codeword $\mathbf{x}$, we form
$\mathbf{x}_{\overline{\mathpzc{E}}}$ as the components of
$\mathbf{x}$ which are indexed by $\overline{\mathpzc{E}}$. And we
arrange all $\mathbf{x}_{\overline{\mathpzc{E}}}$ in the
fingerprinting codebook as rows of a $2^{nR} \times
(n-|\mathpzc{E}|)$ array $\mathbf{X}_{\overline{\mathpzc{E}}}$. The
misidentification happens if $\mathbf{y}_{\overline{\mathpzc{E}}}$
equals to more than two rows of
$\mathbf{X}_{\overline{\mathpzc{E}}}$. With $R>1/2$, $|\mathpzc{E}|
\geq n/2$, and sufficiently large $n$,
\begin{equation} 
2^{nR}-2>2^{(n-|\mathpzc{E}|)}-1.
\end{equation}
And the misidentification will happen with probability at least
$1/3$. From above, we know that if $R>1/2$, the misidentification
probability will be larger than $\delta/3$ for sufficiently large
$n$ which
concludes the proof. \\
\end{proof}

Intuitively, the i.i.d generated codebook will result in
$|\mathpzc{E}| \approx n/2$ number of erased positions with high
probability \cite{Barg_cap}. Then the ``channel'' between one of
the pirates $\mathbf{x}_1$ and the forged fingerprint $\mathbf{y}$
can be approximated by a binary erasure channel (BEC) with erasure
probability 1/2. From \cite{Book_Cover}, we know that the capacity
using the MD decoder of this channel is 1/2. However, in the
two-pirate fingerprinting system, there are always two codewords
$\mathbf{x}_1$ and $\mathbf{x}_2$ in the codebook which meet the
MD decoding criteria. This is the fundamental difference between
this system and the classical BEC channel. In the BEC channel,
with high probability, only one codeword will meet the MD decoding
criteria. As will be presented in Section \ref{Sec_Avg_Simu}, this
difference will have an important implication on the design of
Belief Propagation decoders for fingerprinting. The following
result shows that restricting ourselves to the class of linear
fingerprinting does not entail any performance loss (at least from
an information theoretic perspective) \\
\begin{theorem} \label{Theorem_Avg_linear}
The fingerprinting capacity of the binary linear ensemble with
$t=2$ is $R_{\cal E}=0.5$ (under the averaging attack and the MD
decoder). \\
\end{theorem}
\begin{proof}
We consider the ensemble of  binary linear codes of length $n$ and
dimension $n-l$ defined by the $l \times n $ parity check matrix
$H$, where each entry of $H$ is an
i.i.d Bernoulli random variable with parameter $1/2$. The code rate $R=1-l/n$.\\

\noindent \textbf{Encoder:}  The encoder chooses one codebook from
this linear code ensemble, transfers the fingerprinting codeword
alphabets from $\{0,1\}$ to $\{-1,+1\}$, and assigns the
fingerprints to the users.
\\ \\
\noindent \textbf{Decoder:} With the given forged fingerprint
$\mathbf{y}$, again the decoder treats the position $i$ where
$\mathbf{y}_i=0$ as an erased position, and the others as unerased
positions. The decoder will also transfer the alphabets of
unerased positions from $\{-1,+1\}$ back to $\{0,1\}$. Let
$H_\mathpzc{E}$ denote the submatrix of $H$ that consists of those
columns of $H$ which are indexed by the set of erasures
$\mathpzc{E}$. In a similar manner, let $\mathbf{x}_\mathpzc{E}$
denote those components of the pirate's fingerprint which are
indexed by $\mathpzc{E}$, and
$\mathbf{x}_{\overline{\mathpzc{E}}}$ denote those components
which are indexed by $\overline{\mathpzc{E}}$. In the following,
we assume that the fingerprinting codeword alphabets are
transferred back to $\{0,1\}$ and the addition is module-2. Note
that the true pirates $\mathbf{x}_1$ and $\mathbf{x}_2$ will
result in the same
$\mathbf{x}_{\overline{\mathpzc{E}}}=\mathbf{y}_{\overline{\mathpzc{E}}}$,
where $\mathbf{y}_{\overline{\mathpzc{E}}}$ is defined as in
Theorem \ref{Theorem_Avg_iid}. From the parity check equations,
\begin{equation} \label{Eq_BEC_syndrome}
H_\mathpzc{E}\mathbf{x}^{\mathrm{T}}_\mathpzc{E}=\mathbf{s}^{\mathrm{T}},
\end{equation}
where
$\mathbf{s}^{\mathrm{T}}:=H_{\overline{\mathpzc{E}}}\mathbf{y}^{\mathrm{T}}_{\overline{\mathpzc{E}}}$
is called the syndrome. The syndrome is known at the decoder. The
decoder solves these linear equations to find
$\mathbf{x}_\mathpzc{E}$, combines it with the known
$\mathbf{x}_{\overline{\mathpzc{E}}}=\mathbf{y}_{\overline{\mathpzc{E}}}$,
and declares one of the results as the pirate.
\\ \\
Achievability: We know that (\ref{Eq_BEC_syndrome}) has at least
two solutions corresponding to the true pirates $\mathbf{x}_1$ and
$\mathbf{x}_2$. The rank of $l \times |\mathpzc{E}|$ matrix
$H_\mathpzc{E}$ must equal to $|\mathpzc{E}|-1$ to make sure that
there is only two solutions. The decoder will declare an innocent
user as the pirate if there are more than two solutions, iff
$H_\mathpzc{E}$ has rank less than $|\mathpzc{E}|-1$. This happens
with probability
\begin{equation} \label{Eq_rankless}
1-\frac{M_b(l,|\mathpzc{E}|,|\mathpzc{E}|-1)}{2^{l|\mathpzc{E}|}-M_b(l,|\mathpzc{E}|,|\mathpzc{E}|)},
\end{equation}
where $M_b(l_1,m_1,k_1)$ denote the number of binary matrices with
dimension $l_1 \times m_1$ and rank $k_1$.

To make (\ref{Eq_rankless}) approach zero as $n$ increases, the
second term in (\ref{Eq_rankless}) must approach one as $n$
goes up. To show this, we first assume that $|\mathpzc{E}|+n
\epsilon_1 \leq l$, where $\epsilon_1>0$ is a small number. And
according to (\ref{Eq_NumEminus1}) in Appendix \ref{App_Matrix}
and \cite{di2002finite}, the second term in (\ref{Eq_rankless})
equals
\begin{equation} \label{Eq_ranklessdev1}
\frac{M_b(|\mathpzc{E}|-1,l,|\mathpzc{E}|-1)(2^{|\mathpzc{E}|}-1)}{2^{l|\mathpzc{E}|}-M_b(|\mathpzc{E}|,l,|\mathpzc{E}|)}.
\end{equation}
From \cite{di2002finite}, for $j=0 \ldots |\mathpzc{E}|-1 $
\begin{align} \label{Eq_NumFullRank}
&M_b(|\mathpzc{E}|-j,l,|\mathpzc{E}|-j)=
\prod_{p=0}^{|\mathpzc{E}|-j-1}(2^l-2^p).
\end{align}
Using this formula in (\ref{Eq_ranklessdev1}) and dividing the
nominator and denominator by
$M_b(|\mathpzc{E}|-1,l,|\mathpzc{E}|-1)$, this term equals
\begin{equation} \label{Eq_ranklessdev2}
\frac{2^{|\mathpzc{E}|}-1}{2^{(|\mathpzc{E}|-1)}+2^l[-1+\prod_{p=0}^{|\mathpzc{E}|-2}1/(1-2^{p-l})]}.
\end{equation}
Note that $n\epsilon_1 \leq l-|\mathpzc{E}|$, each $2^{p-l}$
approaches zero exponentially fast with $n$. By using Taylor
series on $1/(1-2^{p-l})$, and with some simplifications, the
denominator becomes
\begin{equation} 
2^{(|\mathpzc{E}|-1)}+\sum_{p=0}^{|\mathpzc{E}|-2}2^{p}+2^{|\mathpzc{E}|}*h.o.t.=2^{|\mathpzc{E}|}*(1+h.o.t.)-1,
\end{equation}
where the higher order terms of the Taylor series are denoted by $h.o.t$ and approach zero exponentially fast. Using this result in
(\ref{Eq_ranklessdev2}), our claim is valid and
(\ref{Eq_rankless}) approaches zero as $n \rightarrow \infty$
if $|\mathpzc{E}|+n \epsilon_1 \leq l$.

As shown in Appendix \ref{App_Theorem_Linear}, $|\mathpzc{E}| \leq
n(1/2+\epsilon)$ with high probability, we know that if
$n(1/2+\epsilon)+n \epsilon_1 \leq l$, or
$R<1/2-(\epsilon+\epsilon_1)$, the probability of
misidentification can be made arbitrary small.
\\ \\
Converse: From (\ref{Eq_linear_dis}) in Appendix, we know that
$P(|\mathpzc{E}| \geq n/2)>P(|\mathpzc{E}|=n/2)=\delta$, where
$\delta$ is non vanishing with respect to codeword length $n$.
With $R>1/2$ and sufficiently large $n$, $P(|\mathpzc{E}|-1 > l)
\geq \delta$. In this case, the rank of $H_{\mathpzc{E}}$ is less
than $|\mathpzc{E}|-1$ and the syndrome decoder will find at least
three solutions of equation (\ref{Eq_BEC_syndrome}). The
misidentification will happen with probability at least $1/3$
since. From above, we know that if $R>1/2$, the probability will
be larger than $\delta/3$ for sufficiently large $n$ and it
concludes the proof. \\
\end{proof}

Next, our approach is generalized to coalitions with $t>2$. The
key to the following corollary  is to treat all alphabets other
than $\pm 1$ in $\mathpzc{A}_{\mathbf{y}}$ of (\ref{Eq_g_avg}) as
erasures.
\\
\begin{corollary} 
The rate $\frac{1}{2^{(t-1)}}$ is MD-achievable for fingerprinting
under
average attack with a coalition of size $t$. \\
\end{corollary}
\begin{proof}
The encoder/decoder are the same as the ones in Theorem
\ref{Theorem_Avg_iid} except for the choices of erasure positions as
described previously. Note that $\mathbf{y}_i \neq \pm 1$ whenever
the pirates' fingerprints bits are not the same at position $i$.
Similar to \cite{Barg_cap}, we know that with high probability, the
i.i.d generated codebooks will meet
\[
|\mathpzc{E}| \leq n\left \{1-\frac{1}{2^{(t-1)}}+\epsilon \right
\}.
\]
Then, following in the footsteps of the proof of Theorem
\ref{Theorem_Avg_iid} we obtain our result.
\end{proof}
\[\;\]

The advantage of the MD decoder, used to obtain the previous result,
is the universality for all $t$. However, for each $t$, we can
obtain higher rates by tailoring our encoder/decoder to this
specific case. To illustaret the idea, let's consider the $t=3$
case. Now, $\mathpzc{A}_{\mathbf{y}}=\{ \pm 1, \pm \frac{1}{3} \}$
and one can achieve better performance by exploiting the information
contained in the positions with $\mathbf{y}_i = \pm \frac{1}{3}$.
\\
\begin{theorem} \label{Theorem_3tJT}
The rate
$H(\frac{1}{8},\frac{1}{8},\frac{3}{8},\frac{3}{8})-H(\frac{1}{4},\frac{1}{2},\frac{1}{4})=0.3113$
is achievable for fingerprinting under average attack with $t=3$.
\\
\end{theorem}

\begin{proof} The encoder is the same as Theorem
\ref{Theorem_Avg_iid}. As for the decoder, we first define $X$ as a random variable with
$P(X=\pm 1)=1/2$, and the random variable $Y=(X+X_2+X_3)/3$, where
$X_2,X_3$ has the same distribution as $X$ and $(X,X_2,X_3)$ are
independent. The transition matrix of $P(Y|X)$ is
\begin{table}[ht]
\begin {center}
\begin{tabular}{c|cccc}
$_{X} \backslash ^{Y} $&-1&-1/3&1/3&1\\ \hline
-1&$\frac{1}{4}$&$\frac{1}{2}$&$\frac{1}{4}$&0\\
1&0&$\frac{1}{4}$&$\frac{1}{2}$&$\frac{1}{4}$. \\
\end{tabular}
\end{center}
\end{table}

Typically, we need a maximum likelihood (ML) decoder designed for
the transition matrix $P(Y|X)$. Note that when $t=2$, this decoder
reduces to the one specified in Theorem \ref{Theorem_Avg_iid}.
However, it is hard to investigate the performance of the ML
decoder, and we use the jointly-typical decoder defined in
\cite{Book_Cover} as a lower-bound for the achievable rate of this
decoder. Given a forged fingerprint $\mathbf{y}$, the decoder
search the codebook to find the codeword such that this codeword
and $\mathbf{y}$ are jointly-typical with respect to $P(X,Y)$ .
Once the decoder finds such a codeword, the decoder declares it as
the pirate.
\\ \\
Achievability : Without loss of generality, we can assume that the
pirates indices are $(1,2,3)$. An event $E_i$ occurs when the
$i$th codeword and $\mathbf{y}$ are jointly typical, and the event
$E^c_i$ is its complement. Then the probability of
misidentification $P_m$ is upper-bounded by
\[
P_m \leq P(E^c_1)+P(E^c_2)+P(E^c_3)+\sum_{i \neq 1,2,3} P(E_i).
\]

From \cite[Theorem 15.2.1]{Book_Cover}, the first three terms can be made less than any arbitrary small $\epsilon>0$ for sufficiently large $n$. And
the last term is upper-bounded by
\[
(M-3) 2^{-n(I(X;Y)-4\epsilon)},
\]
So if $R<I(X;Y)-4\epsilon$, $P_m$ can be made arbitrary small for
sufficiently large $n$. According to the transition matrix of
$P(Y|X)$, we know that
\[
I(X;Y)=H(\frac{1}{8},\frac{1}{8},\frac{3}{8},\frac{3}{8})-H(\frac{1}{4},\frac{1}{2},\frac{1}{4}),
\]
which concludes the proof
\end{proof}

\section{The Marking Assumption} \label{Sec_Theorem_2}
Having studied the special case of averaging attack, we now
proceed to the case when the coalition can employ \emph{any}
strategy as long as the marking assumption is satisfied. The
following result establishes the achievable rate of random
fingerprinting codes with MD decoding \\
\begin{theorem} \label{Theorem_Random}
For all rates less than $1-H(0.25)$ there exists an MD-achievable
fingerprinting code, when $t=2$. \\
\end{theorem}
\begin{proof}
We use a random coding argument to prove our result. We construct
the following ensemble of binary random codes as in Theorem
\ref{Theorem_Avg_iid}: Binary random vectors (fingerprints) of
length $n$ are assigned to the $M=2^{nR}$ users where each
coordinate is chosen independently with equal probability of being
$0,1$. For a small $\varepsilon$, we say the assigned fingerprints
$\mathbf{x}_1,\mathbf{x}_2$ are close if
$d_H(\mathbf{x}_1,\mathbf{x}_2)\leq n(\frac{1}{2}+\varepsilon)$.
If the pair $(\mathbf{x}_1,\mathbf{x}_2)$ is close we denote it by
$\mathbf{x}_1\overset{C}{\leftrightarrow} \mathbf{x}_2$, otherwise
for a non-close pair we write:
$\mathbf{x}_1\overset{N}{\leftrightarrow} \mathbf{x}_2$. Given a
forged fingerprint $\mathbf{y}$, the average probability of
misidentification over this ensemble can be upper bounded by:
\[
P_m(\mathbf{y}|\mathbf{x}_1 \overset{C}{\leftrightarrow}
\mathbf{x}_2)+P(\mathbf{x}_1 \overset{N}{\leftrightarrow}
\mathbf{x}_2),
\]
where $P_m(\mathbf{y}|\mathbf{x}_1
\overset{C}{\leftrightarrow}\mathbf{x}_2)$ is the
misidentification probability when $\mathbf{y}$ is produced by a
close pair $(\mathbf{x}_1,\mathbf{x}_2)$ and $P(\mathbf{x}_1
\overset{N}{\leftrightarrow} \mathbf{x}_2)$ is the probability
that the pirates did not constitute a close pair. Both probability
are averaged over the random coding ensemble. By the following
argument, we will show that these probabilities goes exponentially
to zero as $n$ goes to infinity hence the proof.

In Appendix \ref{App_close_pair} we have proved that
$P(\mathbf{x}_1 \overset{N}{\leftrightarrow} \mathbf{x}_2)$ goes
to zero as $n$ goes to infinity. Now we turn to
$P_m(\mathbf{y}|\mathbf{x}_1 \overset{C}{\leftrightarrow}
\mathbf{x}_2)$. Since
$d_H(\mathbf{x}_1,\mathbf{x}_2)<n(\frac{1}{2} + \varepsilon )$,
the Hamming distance of the forged copy $\mathbf{y}$ with at least
one of the pirates must be less than $h(n) := n(\frac{1}{4} +
\frac{\varepsilon}{2})$ due to the marking assumption. Without
loss of generality, we assume this pirate to be $\mathbf{x}_1$.
Using minimum Hamming distance decoding, misidentification occurs
if there is another binary vector $\mathbf{z}$ of length $n$ in
the codebook such that $d_H(\mathbf{y},\mathbf{z}) \leq
d_H(\mathbf{y},\mathbf{x}_1)$. The total probability of this event
in the random ensemble is upper-bounded by
\[
\frac{M\sum_{i=1}^{h(n)}{n\choose i}}{2^n}\doteq
M*2^{-n(1-H(0.25))},
\]
where the union bound is used. The probability of
misidentification in a random code of size $M=2^{nR}$ is at most
\[
2^{-n(1-H(0.25)-R)}.
\]
The above probability goes exponentially to zero as $n\rightarrow
\infty$ for all rates $R<1-H(0.25)$.  $\;$ \\
\end{proof}
Intuitively, with a high probability, the forged copy will be
produced by a pair of {\em close} pirates. Therefore, the minimum
Hamming distance between the pirates $\mathbf{x}_1$ and the forged
copy $\mathbf{y}$ is approximately $n/4$ implying that we can
treat the "channel" between them as a binary symmetric channel
(BSC) with crossover probability $1/4$ (whose capacity is
$1-H(0.25)$ \cite{Book_Cover}). Next, we extend our result to
binary linear codes \\

\begin{theorem} \label{Theorem_Linear}
For all rates less than $1-H(0.25)$, there exists a \emph{linear}
MD-achievable fingerprinting code, when $t=2$.
\\
\end{theorem}
\begin{proof}
Consider the ensemble of binary linear codes with binary parity
generator matrix $G$ where elements of $G$ are chosen equally and
independently from $\{0,1\}$ similar to Theorem
\ref{Theorem_Avg_linear}. The size of matrix $G$ is $(n-l) \times
n$, with rate $R=(n-l)/n$ and the codeword length $n$. It should
also be noted that in the following all matrix multiplications and
additions are done in module-2 unless otherwise stated. In order
to randomize the codebook, the distributor employs the following
strategy: Generating the secret key vectors as independent binary
random vectors of length $n$, whose coordinates are chosen to be
$0,1$ independently with probability 1/2. We denote the vector
indexed by secret key $k$ as $\mathbf{k}$. The vector $\mathbf{k}$
is added in the binary domain to the codeword, and the resulting
vector is assigned to the corresponding user. Note that this
operation will not change the detectable positions, where the
codewords are the different. With forged copy $\mathbf{y}$, the
decoder subtracts $\mathbf{k}$ and performs MD decoding. As we
mentioned earlier, the secret key is unknown to the users and is
only known to the distributor.

Similar to the proof of Theorem \ref{Theorem_Random}, we can
upper-bound the probability of misidentification as
\begin{equation}
P_m(\mathbf{y}|\mathbf{x}_1 \overset{C}{\leftrightarrow}
\mathbf{x}_2)+P(\mathbf{x}_1 \overset{N}{\leftrightarrow}
\mathbf{x}_2).
\end{equation}
In Appendix \ref{App_Theorem_Linear} we have established that over
the ensemble of linear random codes described above,
$P(\mathbf{x}_1 \overset{N}{\leftrightarrow} \mathbf{x}_2)$ also
goes to zero as the code length goes to infinity. Now let us
consider $P_m(\mathbf{y}|\mathbf{x}_1 \overset{C}{\leftrightarrow}
\mathbf{x}_2).$ The codes assigned to the users which are the result of the addition of a secret key to a linear code can be written as:
\begin{equation} \label{coset}
\mathbf{u}G + \mathbf{k}
\end{equation}
where $\mathbf{u}$ is an information message vector. Notice that the ensemble defined by (\ref{coset}) is the same as ensemble of \emph{coset codes} introduced in \cite{Gallager}. In our proof, we need the following lemmas for the coset codes ensemble that are proved in \cite{Gallager}.
\newtheorem{lemma 1}{Lemma}
\begin{lemma} \label{Lm1}
The probability of any binary vector $\mathbf{v}$ being a codeword in the ensemble defined by (\ref{coset}) is equal to $2^{-n}$.
\end {lemma}
\begin{lemma} \label{Lm2}
Let $\mathbf{v_1}$, $\mathbf{v_2}$ be the codewords corresponding to two different information sequences $\mathbf{u_1}$, $\mathbf{u_2}$. Then over the ensemble of codes, $\mathbf{v_1}$, $\mathbf{v_2}$ are statistically independent.
\end {lemma}
Similar to the proof of Theorem \ref{Theorem_Random}, again due to the marking assumption we can
assume $d_H(\mathbf{y,x}_1) < h(n).$ Using MD decoding, misidentification occurs
if there is another binary vector $\mathbf{z}$ of length $n$ in
the codebook such that $d_H(\mathbf{y},\mathbf{z}) \leq
d_H(\mathbf{y},\mathbf{x}_1)$. The total number of binary vectors for which $d_H(\mathbf{y},\mathbf{z}) \leq
d_H(\mathbf{y},\mathbf{x}_1)$ can be upper bounded by: $\sum_{i=1}^{h(n)}{n\choose i} \doteq 2^{nH(0.25)}$. By Lemma \ref{Lm1} and Lemma \ref{Lm2} over the ensemble each of such vectors $\mathbf{z}$ is independent of $\mathbf{x_1}$ with probability $2^{-n}$. Therefore, the total probability of this event in the ensemble is upper-bounded by:
\[
M*2^{-n(1-H(0.25))},
\]
where again the union bound is used. The probability of
misidentification in a random coset code of size $M=2^{nR}$ is at most
\[
2^{-n(1-H(0.25)-R)}.
\]
The above probability goes exponentially to zero as $n\rightarrow
\infty$ for all rates $R<1-H(0.25)$.  $\;$ \\
\end{proof}

When the coalition size, $t$ is larger than two, the minimum
distance decoding will fail due to the following argument. Let
$t=3$ and assume that the forged copy is produced by
\[
\mathbf{y=x}_1 + \mathbf{x}_2 +  \mathbf{x}_{3},
\]
where the additions are modulo-2. It is easy to check that this
attack satisfies the marking assumption. For $t > 3$ the coalition
can consider only three of the pirates, ignore the rest and apply
this attack. Following the footsteps in the proof of Theorem
\ref{Theorem_3tJT}, it is easy to see that the MD-achievable rate
is zero. Indeed, it can also be shown that the resulting ``BSC
channel" has crossover probability $1/2$, and this negative result
is obtained \cite{Book_Cover}.

\section{Belief Propagation for Fingerprinting} \label{Sec_BP} Implementing the exact minimum distance decoder may require
prohibitive complexity (especially for large codeword lengths).
This motivates our approach of using the BP framework to
approximate the MD decoder. More specifically, in this section, we
present explicit constructing of graph-based codes, along with the
corresponding BP decoders, which are tailored for the
fingerprinting application.
\subsection{Averaging attack}
\label{Sec_Avg_Simu} As remarked earlier, the two-pirate averaging
attack will produce a ``channel'' {\em almost equivalent} to the
classical BEC. This inspires the use of graphical codes based on
the Repeat Accumulate (RA) framework \cite{divsalar1998ctt}, such
as the nonsystematic irregular RA code of \cite{pfister2005cae}
and the irregular ARA code of \cite{Pfister2007ARA}, which were
shown to be capacity achieving for the BEC. In our simulations, we
use the original regular RA codes of \cite{divsalar1998ctt} due to
their simplicity and good performance for low rate scenarios. It
is worth noting that all the techniques discussed in the sequel
can be applied directly to the irregular codes presented in
\cite{pfister2005cae,Pfister2007ARA}. For the sake of completeness
we review briefly the encoding procedure for regular RA codes:
first, the information bits are repeated a constant number of
times (by a regular repetition code) and interleaved. The
interleaved bits are then accumulated to generate the code
symbols. Similarly, one can employ the standard BP iterative
decoding approach \cite{luby2001eec} to identify the pirates.
However, we argue next that significant performance improvement
can be obtained via a key modification to the iterative
decoder\footnote{in the following, the fingerprinting codeword
alphabets are $\{0,1\}$ after decoder transformation and the
addition is module-2.}.

It is well known that the standard iterative algorithm will fail if
a stopping set exists in the erased positions \cite{di2002finite}.
Unfortunately, a stopping set always will exist in the erased
positions produced by averaging attack. To see this, it is more
convenient to represent the RA code using the appropriate bipartite
Tanner graph containing a set of variables
$\mathpzc{V}=\{v_1,v_2,\ldots\}$ and a set of check nodes. The
reader are referred to
\cite{divsalar1998ctt,pfister2005cae,Pfister2007ARA} for more
details on the graphical representation of RA codes. A stopping set
$\mathpzc{S}$ is, therefore, a subset of $\mathpzc{V}$, such that
all neighbors of $\mathpzc{S}$ are connected to $\mathpzc{S}$ at
least twice. The standard BP algorithm
can now be stated as the follows. \\

\noindent \textit{[Standard BP]}:
\begin{enumerate}
  \item Find a check node that satisfies the following
        \begin{itemize}
          \item This check node is not labelled as ``finished''.
          \item The values of all but one of the variable nodes connected to the check node are known.
        \end{itemize}
        Set the value of the
  unknown erased one to be the module-2 addition of the other variable nodes. And label
  that check node as ``finished''.
  \item \textit{Repeat} step 1 until all check nodes are labeled as ``finished''
  or the decoding cannot continue further. If the latter happens,
  declare the decoding fail.
  \\
\end{enumerate}
It is now easy to see that, in the stopping set, every check node
is connected to at least two erased variable nodes and the decoder
will halt at this point. The following result establishes the
limitation of the standard BP decoder in our fingerprinting
scenario \\

\begin{proposition} 
Let $\mathpzc{V}_{B1}$ and $\mathpzc{V}_{B2}$ be the set of values
of the variable node set $\mathpzc{V}$ corresponding to pirate
fingerprints $\mathbf{x}_1$ and $\mathbf{x}_2$, respectively. And
let $\mathpzc{V}_d$ be the set of variable nodes where the
corresponding values in $\mathpzc{V}_{B1}$ and $\mathpzc{V}_{B2}$
are different. Then $\mathpzc{V}_d$ is a stopping set.
\\
\end{proposition}

\begin{proof}
This statement is proved by contradiction. First we assume that
$\mathpzc{V}_d$ is not a stopping set. It means that $\exists j \in
\bigcup_{i \in \mathpzc{V}_d}N(i)$ where the check node $j$ has only
one neighbor $i'$ in $\mathpzc{V}_d$. Here we denote the neighbor of
node $i$ in the graph as $N(i)$. For the neighboring variable nodes
of this check node, we have
\begin{equation} \label{Eq_Stop_dev1}
\left\{
\begin{array}{ll}
\mathpzc{V}_{B1}(i)=\mathpzc{V}_{B2}(i) & \forall i \neq i', i \in N(j) \\
\mathpzc{V}_{B1}(i')=\mathpzc{V}_{B2}(i')+1 & i' \in N(j).  \\
\end{array} \right.
\end{equation}
However, from the check equation of this check node
\begin{equation} \label{Eq_Stop_dev2}
\sum_{i \in  N(j)} \mathpzc{V}_{B1}(i)=\sum_{i \in N(j)}
\mathpzc{V}_{B2}(i)=0,
\end{equation}
where the addition is module-2. It is obvious that
(\ref{Eq_Stop_dev1}) contradicts with (\ref{Eq_Stop_dev2}) since
the total number of variable nodes such that $\mathpzc{V}_{B1}(i)
\neq \mathpzc{V}_{B2}(i), i \in N(j)$ should be even. Thus
$\mathpzc{V}_d$ is a stopping set. \\
\end{proof}
Since, under averaging attack, the bits of the forged fingerprint
will be erased whenever the pirate fingerprints are different in
the Tanner graph, then $\mathpzc{V}_d$ will be always contained in
the erased positions and the iterative decoder will fail. The
modification, presented next, will break the stopping set
$\mathpzc{V}_d$, and hence, allow the iterative decoder to proceed
forward. The key observation is that for every erased position in
$\mathpzc{V}_d$, the pirate fingerprints can only be represented
by only two combinations $\{0,1\}$ or $\{1,0\}$. It allows us to
choose one variable node in this stopping set, and set its value
to $1$. The modified forged
fingerprint will then be ``closer'' to one of the pirate fingerprints. In summary, the decoder becomes\\

\noindent \textit{[Modified BP for fingerprinting]}:

\begin{enumerate}
\item Perform the standard BP algorithm, remove all the ``finished'' labels and \textit{Go to} step $2$
\item Choose a proper variable node in $\mathpzc{V}_d$ (different from previous choices), and set
  its value to $1$. If the decoder has executed this step more
  than $N_{max}$ times, declare a decoding failure and exit.
  \item Run the standard BP on the new graph. If the decoder fails,
  reset the variable nodes to their original values and \textit{Go
  to} step $2$.
  \\
\end{enumerate}

\begin{figure}
\centering \epsfig{file=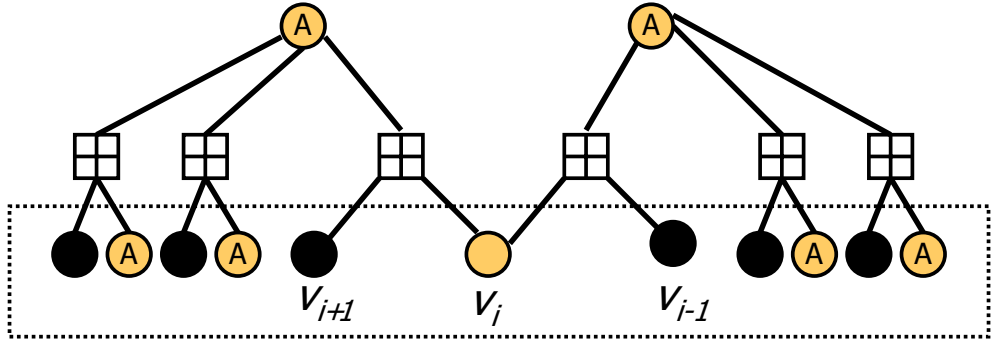 ,
width=0.8\textwidth} \caption{Proper variable node to be chosen in
step 2 of proposed modified BP algorithm for two-pirate averaging
attack.} \label{Fig_var_choose}
\end{figure}

In step $2$, we must make sure that the chosen variable node
breaks the stopping set $\mathpzc{V}_d$. The neighboring variables
nodes of a degree-3 check node in RA code are good choices. From
the check equations in (\ref{Eq_Stop_dev2}), the erased variables
nodes will appear in pair. If we set the value of one of the two
erased neighbor variable node $v_i$ as 1, this degree-3 check node
is connected to $\mathpzc{V}_d \setminus v_i$ with only one edge.
Then $\mathpzc{V}_d \setminus v_i$ is not a stopping set. We also
need to choose the variable node which will affect as much other
variable nodes in $\mathpzc{V}_d \setminus v_i$ as possible by
setting its value. Since all check nodes of RA code are degree-3,
we choose such variable node $v_i$ in the degree-2
variable-node-chain of RA code, as shown in Fig.
\ref{Fig_var_choose}. The check node is depicted as $\boxplus$,
the unerased variable nodes as black circle and the erased ones as
hollow circle. Furthermore, each variable node which will benefit
from guessing $v_{i}$ is shown as hollow circle with the letter
``A'' in the figure. The key observation is that, for node $v_i$,
the two neighboring accumulator output nodes, i.e., $v_{i-1}$ and
$v_{i+1}$, correspond to non-erased bits. This implies that that
setting the value of $v_i$ will at least affect $6$ other variable
nodes of rate $1/3$ RA code.

\begin{figure}
\centering \epsfig{file=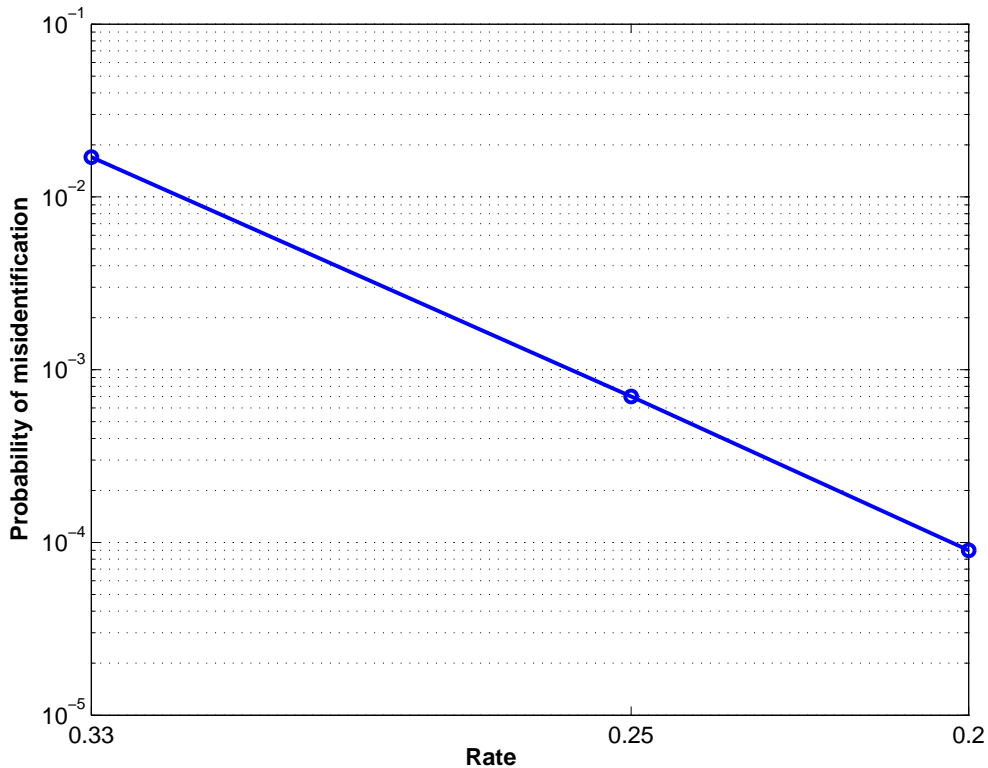 ,
width=0.8\textwidth} \caption{Probability of misidentification
under two-pirate averaging attack using RA codes with different
rates and modified BP algorithm without variable node selection.}
\label{Fig_FP_RA_BEC_Rate}
\end{figure}

\begin{figure}
\centering \epsfig{file=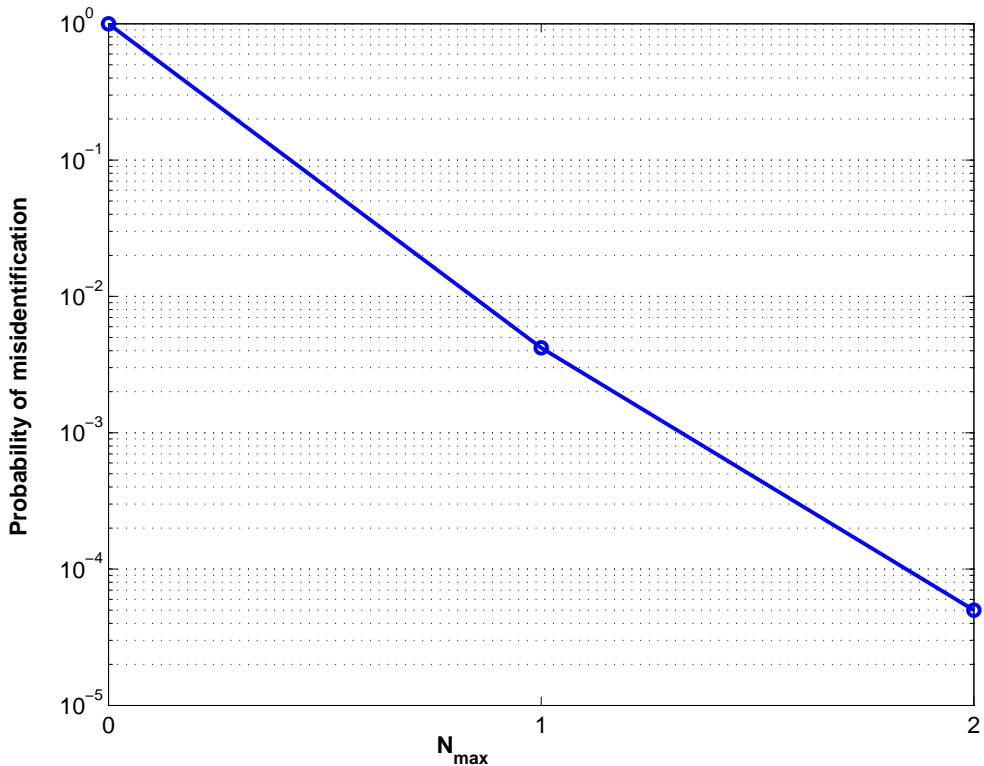 ,
width=0.8\textwidth} \caption{Probability of misidentification
under two-pirate averaging attack using rate 1/3 RA code and
modified BP algorithm with different $N_{max}$.}
\label{Fig_FP_RA_BEC}
\end{figure}

\begin{figure}
\centering \epsfig{file=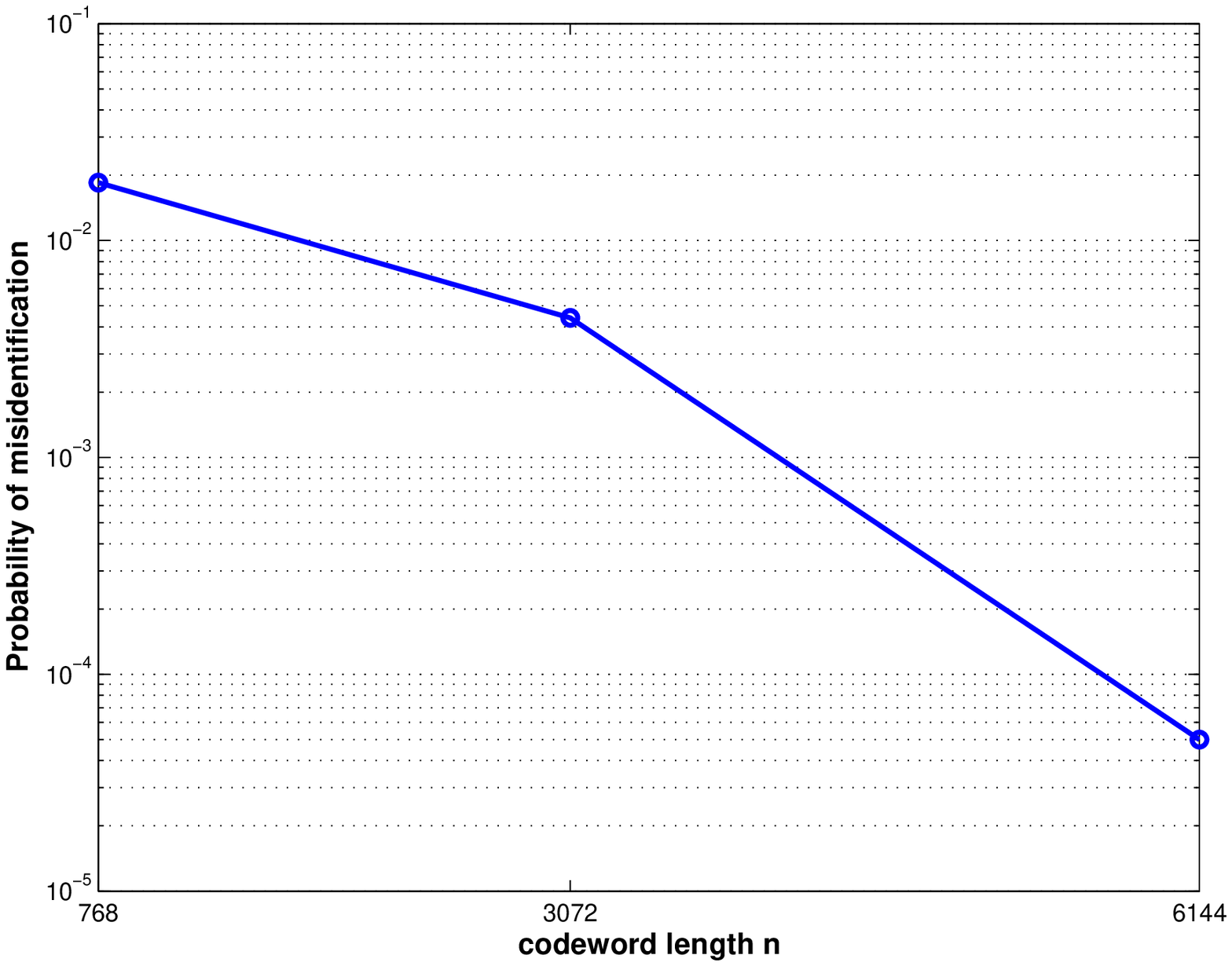,
width=0.8\textwidth} \caption{Probability of misidentification
under two-pirate averaging attack using rate 1/3 RA code and
modified BP algorithm with different code lengths $n$.}
\label{Fig_FP_RA_BEC_Len}
\end{figure}

Now, we are ready to report our simulation results. First, we show
the performance of proposed algorithm with different rate RA codes
without variable node selection in Fig~\ref{Fig_FP_RA_BEC_Rate}
(i.e., we select the first unerased variable node in the RA
degree-2 variable-node-chain and set $N_{max}=1$). Here, the
number of information bits $n/R=16384$ is fixed for all rates, to
make the number of users $M$ the same. We observe that, without
selecting the variable node as shown in Fig~\ref{Fig_var_choose},
the probability of misidentification $\bar{P}^a_m$ is high for
rate $1/3$. This performance can be improved by the proposed
algorithm for variable node selection and increasing $N_{max}$ as
depicted in Fig.~\ref{Fig_FP_RA_BEC}. Finally, in
Fig.~\ref{Fig_FP_RA_BEC_Len} we report $\bar{P}^a_m$ with
different code length $n$ and $N_{max}=2$.

Finally, we note that our algorithm is similar, in spirit, to the
proposed guessing algorithm in \cite{pishronik2004dld}. The critical
difference is that the structure of our problem ensures that the
guessed bit always corresponds to one of the pirates, and hence, we
do not need to worry about the possibility of contradictions as the
iteration proceeds.
\subsection{The Marking Assumption: The Memoryless Attack}
In this subsection, we report our simulation results for the
two-pirate memoryless attack. In this attack, when the pirates
encounter a detectable position, they choose $0,1$ independently
and with equal probability to form the forged copy. We use rates
$1/8$, $1/9$ and $1/10$ ARA codes based on the low rate
protographs presented in \cite{divsalar}. The protographs of the
codes are depicted in Fig~\ref{Fig_Prot}. For a formal description
of the ARA codes, we refer the interested readers to
\cite{divsalar}, \cite{Pfister2007ARA} and references therein.
Decoding is done iteratively using the BP framework with a maximum
number of iterations equal to $60$. Here, the decoder treats the
forged fingerprint as the output of a BSC with crossover
probability equal to $0.25$. In Fig~\ref{Pm_Fig}, the probability
of misidentification $\bar{P}_m$ is depicted versus different code
lengths for different rates. As shown in the figure, it is clear a
vanishing small misidentification probability is achievable for
rate $1/9$ which is about an order of magnitude higher than the
best result available in the literature for explicit
fingerprinting codes.
\begin{figure}
\centering \epsfig{file=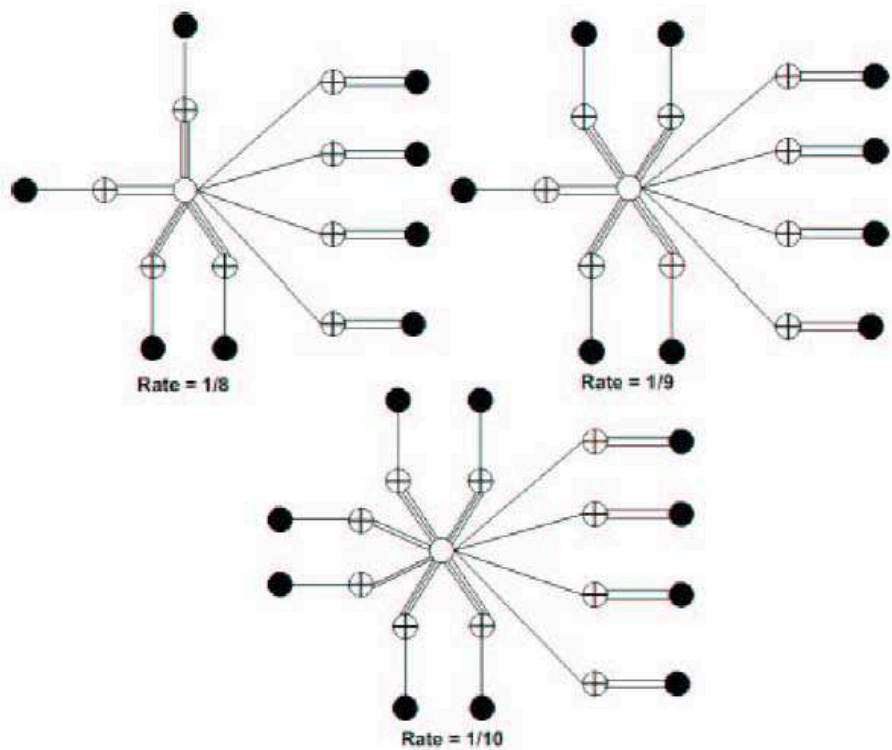 , width=0.8\textwidth}
\caption{Protographs of rate 1/8, 1/9, 1/10 ARA codes.}
\label{Fig_Prot}
\end{figure}

\begin{figure}
\centering \epsfig{file=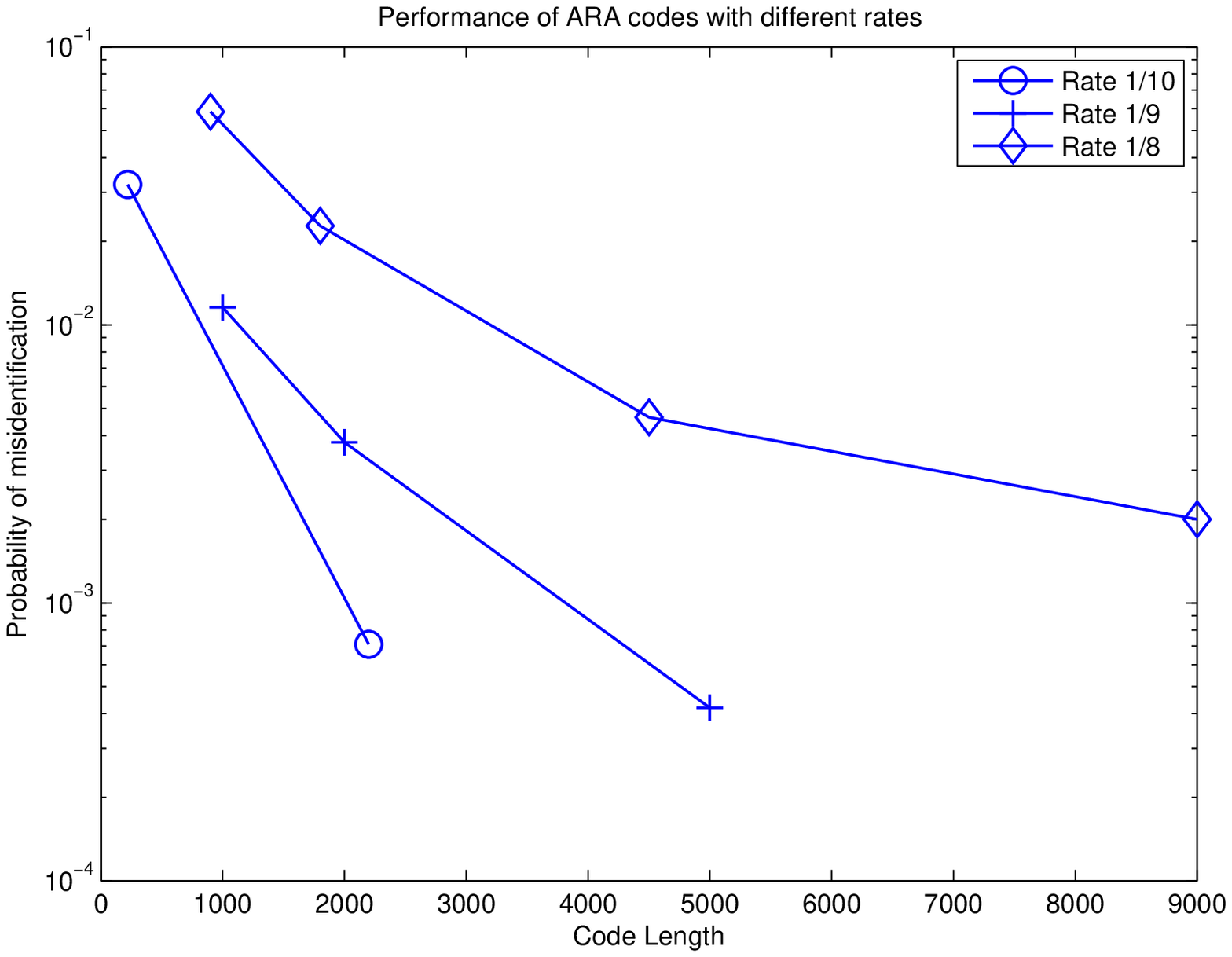 ,
width=0.8\textwidth} \caption{Probability of misidentification for
ARA codes with different rates and code lengths, under two-pirate
memoryless attack.} \label{Pm_Fig}
\end{figure}

\section{Conclusion} \label{Sec_conclu}
This paper developed an information theoretic framework for the
design of low complexity coding/decoding techniques for
fingerprinting. More specifically, we established the superior
performance of the minimum distance decoder and validated our
theoretical claims via explicit construction of BP
encoding/decoding schemes. In the averaging attack scenario, our
framework was inspired by the equivalence between our problem and
the BEC. We also showed that the worst case attack, under the
marking assumption, is equivalent to a BSC with a cross-over
probability equal to $1/4$. Our approach for the averaging attack
can handle arbitrary coalition sizes, whereas it was shown that
the MD decoder recover from marking assumption attacks only with
coalitions composed of two pirates. This negative result motives
our current investigations on more sophisticated approaches for
pirate tracing using the intimate connection between collusion in
digital fingerprinting and multiple access channels.

\appendices
\section{On non-close pairs in random ensemble}
We will examine the probability of non-close pairs for random
i.i.d and linear codebook ensembles, and show that these events
will not happen with high probability.
\subsection{i.i.d codebook ensemble} \label{App_close_pair}
For a codebook $C$ in the i.i.d ensemble and $1\leq d\leq n$,
define the number of unordered pairs of codewords
$(\mathbf{x}_i,\mathbf{x}_j)$ with $i \neq j$ in $C$ at distance
$d$ apart as
\begin{equation}
S_c(d):=\sum_{i=1}^{M} \sum_{j=1}^{i-1} \Phi \{
d_{H}(\mathbf{x}_i,\mathbf{x}_j)=d\},
\end{equation}
where $\Phi (\cdot )$ is the indicator function. In \cite{Forney},
it is established that with probability going to one as
$n\rightarrow \infty $
\begin{equation} \label{Eq_random_dis}
S_c(d) \doteq  \left\{ \begin{array}{ll}
2^{n(2R+H(\frac{d}{n})-1)} & \textrm{ $n\delta_{GV}(2R)<d<n(1-\delta_{GV}(2R)) $}\\
0 & \textrm{otherwise,}\\
\end{array} \right.
\end{equation}
where $\delta_{GV}(\cdot) $ is the Gilbert-Varshamov distance
which for $0<R<1$, $\delta_{GV}(R) $ is defined as the root
$\delta<0.5$ of the equation $H(\delta)=1-R$. And $\delta_{GV}(R)
$ is zero for $R \geq 1$. Using (\ref{Eq_random_dis}), we can
write the probability of non-close pairs in the codes of the
random ensemble as
\begin{equation}
\frac{\sum_{d>n(1/2+\epsilon)}^{n(1-\delta_{GV}(2R))}2^{n(2R+H(d/n)-1)}}{2^{2nR}}<\frac{n2^{n(2R-1+H(\frac{1}{2}+\epsilon))}}{2^{2nR}}.
\end{equation}
which goes exponentially to zero as $n \rightarrow \infty.$
\subsection{Random binary linear codebook ensemble} \label{App_Theorem_Linear}
For a code $C$ in the linear ensemble and $1 \leq  d \leq n$ by
the symmetry of linear codes we can write
\begin{equation}
S_c(d) = \sum_{i=1}^{M} \sum_{j=1}^{i-1} \Phi \{
d_{H}(\mathbf{x}_i,\mathbf{x}_j)=d\}=\frac{1}{2}\sum_{i=1}^{M}
\sum_{j \neq i} \Phi \{
d_{H}(\mathbf{x}_i,\mathbf{x}_j)=d\}=\frac{M}{2}N_c(d)\doteq
2^{nR} N_c(d),
\end{equation}
where $N_c(d) := \sum_{j \neq i} \Phi \{
d_{H}(\mathbf{x}_i,\mathbf{x}_j)=d\}$. In \cite{Forney},  it is
shown that with probability going to one as $n \rightarrow \infty$
\begin{equation} \label{Eq_linear_dis}
N_c(d) \doteq \{
\begin{array}{cc} 2^{n(R+H(d/n)-1)} , &n
\delta_{GV}(R) <d<n(1-\delta_{GV}(R)) \\ 0 , & \; otherwise.
\end{array}
\end{equation}
Therefore, the average probability of a pair being non-close can
be written as
\begin{equation}
\frac{\sum_{d>n(1/2+\epsilon)}^{n(1-\delta_{GV}(R))}2^{n(2R+H(d/n)-1)}}{2^{2nR}}<\frac{n2^{n(2R-1+H(\frac{1}{2}+\epsilon))}}{2^{2nR}},
\end{equation}
which again goes exponentially to zero as $n \rightarrow \infty.$

\section{Computation of $M_b(l,|\mathpzc{E}|,|\mathpzc{E}|-1)$} \label{App_Matrix}
We will show that for $l \geq |\mathpzc{E}|$
\begin{equation} \label{Eq_NumEminus1}
M_b(l,|\mathpzc{E}|,|\mathpzc{E}|-1)=M_b(|\mathpzc{E}|-1,l,|\mathpzc{E}|-1)(2^{|\mathpzc{E}|}-1).
\end{equation}

To this end, by symmetry,
\[
M_b(l,|\mathpzc{E}|,|\mathpzc{E}|-1)=M_b(|\mathpzc{E}|,l,|\mathpzc{E}|-1).
\]
And from Appendix A of \cite{di2002finite} and $|\mathpzc{E}| \leq
l$, the RHS equals to
\begin{align} \notag
&M_b(|\mathpzc{E}|,l,|\mathpzc{E}|-1)=M_b(|\mathpzc{E}|-1,l,|\mathpzc{E}|-1)2^{|\mathpzc{E}|-1}\\
\notag&+M_b(|\mathpzc{E}|-1,l,|\mathpzc{E}|-2)(2^{l}-2^{|\mathpzc{E}|-2}).
\notag
\end{align}
From Appendix A of \cite{di2002finite}, we also have the following
recursive formula for $j=1 \ldots |\mathpzc{E}|-2 $
\begin{align} \notag
&M_b(|\mathpzc{E}|-j,l,|\mathpzc{E}|-1-j)=M_b(|\mathpzc{E}|-1-j,l,|\mathpzc{E}|-1-j)2^{|\mathpzc{E}|-1-j}\\
\notag&+M_b(|\mathpzc{E}|-1-j,l,|\mathpzc{E}|-2-j)(2^{l}-2^{|\mathpzc{E}|-2-j}).
\end{align}
And $M_b(|\mathpzc{E}|,l,|\mathpzc{E}|-1)$ equals to
\begin{align} \label{Eq_NumEminus1_dev_1}
\sum_{j=1}^{|\mathpzc{E}|-1} & \left \{
M_b(|\mathpzc{E}|-j,l,|\mathpzc{E}|-j)2^{|\mathpzc{E}|-j}
\prod_{p=1}^{j-1}(2^l-2^{|\mathpzc{E}|-1-p}) \right \} \\ \notag +
& M_b(1,l,0)*(2^l-1)
\prod_{p=1}^{|\mathpzc{E}|-2}(2^l-2^{|\mathpzc{E}|-1-p}),
\end{align}
where $M_b(1,l,0)=1$.

Finally, using (\ref{Eq_NumFullRank}) in
(\ref{Eq_NumEminus1_dev_1}),
\[
M_b(|\mathpzc{E}|,l,|\mathpzc{E}|-1)=\sum_{j=1}^{|\mathpzc{E}|}M_b(|\mathpzc{E}|-1,l,|\mathpzc{E}|-1)2^{|\mathpzc{E}|-j},
\]
And it is easy to check that the above formula equals to
(\ref{Eq_NumEminus1}).

\bibliographystyle{IEEEtran}
\bibliography{FP}

\end{document}